\documentclass[floatfix,amsmath,amssymb,amsfonts,bbm,onecolumn,superscriptaddress,letterpaper,nofootinbib]{revtex4}
\usepackage[latin1]{inputenc}
\usepackage{amsmath,hyperref}
\usepackage{amsfonts}
\usepackage{bbm}
\usepackage{verbatim}
\usepackage[pdftex]{graphicx}

\newcommand{\ket}[1]{|#1\rangle}
\newcommand{\bra}[1]{\langle #1|}

\newtheorem{theorem}{Theorem}
\newtheorem{lemma}{Lemma}

\newcommand{\tr}{\text{Tr}}
\newcommand{\Tr}{\text{Tr}}
\newcommand{\half}{\mbox{$\textstyle \frac{1}{2}$}}

\newcommand{\be}{\begin{eqnarray}}
\newcommand{\ee}{\end{eqnarray}}
\newcommand{\ketL}[1]{|#1\rangle\rangle}

\newenvironment{proof}[1][Proof]{\noindent\textbf{#1.} }{\ \rule{0.5em}{0.5em}}

\begin{document}

\author{J.-C. Boileau}
\affiliation{Perimeter Institute for Theoretical Physics, 31 Caroline Street North, Waterloo, ON, N2L 2Y5, Canada}
\affiliation{Institute for Quantum Computing, University of Waterloo, Waterloo, ON, N2L 3G1, Canada}
\author{L. Sheridan}
\affiliation{Institute for Quantum Computing, University of Waterloo, Waterloo, ON, N2L 3G1, Canada}
\author{M. Laforest}
\affiliation{Institute for Quantum Computing, University of Waterloo, Waterloo, ON, N2L 3G1, Canada}
\author{S. D. Bartlett}
\affiliation{School of Physics, The University of Sydney,  Sydney, New South Wales 2006, Australia}

\title{Quantum Reference Frames and the Classification of Rotationally-Invariant Maps}

\date{January 14$^{th}$ 2008}

\begin{abstract}%change abstract
%We analyze, in general, the evolution of a quantum directional reference frame when it is used as a resource for performing quantum operations. We introduce the \emph{moments} of a quantum reference frame, which serve as a complete description of its properties as a frame. We give a convenient representation for any map which is covariant with respect to an irreducible representation of $SU(2)$. We use these results to investigate how many times a quantum directional reference frame represented by a spin-$j$ system can be used to perform a certain quantum operation with a given probability of success. From our results follows a classification of the dynamics of spin-$j$ system under the repeated action of any covariant map with respect to $SU(2)$.
We give a convenient representation for any map that is covariant with respect to an irreducible representation of $SU(2)$, and use this representation to analyze the evolution of a quantum directional reference frame when it is exploited as a resource for performing quantum operations. We introduce the \emph{moments} of a quantum reference frame, which serve as a complete description of its properties as a frame, and investigate how many times a quantum directional reference frame represented by a spin-$j$ system can be used to perform a certain quantum operation with a given probability of success. We provide a considerable generalization of previous results on the degradation of a reference frame, from which follows a classification of the dynamics of spin-$j$ system under the repeated action of any covariant map with respect to $SU(2)$.

\end{abstract}

\maketitle

\section{Introduction}
\label{sec:intro}

Quantum operations, such as the application of a unitary rotation on a qubit or a projective measurement in some basis, require some form of reference frame.  In most descriptions of quantum experiments, this reference frame is considered to be classical and therefore can be treated as non-dynamical and used repeatedly without disturbance.  The situation is quite different for a quantum reference frame -- measurements and interactions can cause an unknown disturbance on the quantum frame and, consequently, it can reduce the ability of that system to serve as a reference frame in subsequent uses~\cite{GI05, BIM05, brst,py,brs}.%add two references

Consider, for example, an apparatus implementing a unitary rotation operation on a qubit.  For concreteness, we consider this qubit to be spin-1/2 system; however, it could equally well be the polarization state of a single photon, a two-level atom, etc.  A standard apparatus uses a classical directional reference frame to define the axis about which the rotation occurs.  Now suppose that due to miniaturization, space constraints, or to improve the speed of the operation, the system that serves as the reference direction is not accurately described classically, but instead requires a quantum treatment using a Hilbert space of finite size.  The quantum operation describing the unitary rotation is then an operation \emph{conditional} on the state of the quantum reference frame.  Two effects will occur as a result of the quantum nature of the reference frame.  First, this conditional operation will not be identical to the classical case due to the inherent uncertainty in the direction of the quantum reference frame.  Second, upon use, the state of the quantum reference frame may experience an unknown disturbance, essentially due to entanglement induced between the reference frame and system as a result of the interaction.

%new paragraph
It should be noted that the constraints that lead to the requirement that the reference system be treated quantumly are ones that are expected to arise in a quantum computing architecture.  For example, measuring devices for such a system must couple strongly to specific registers (subsystems), and this could imply a very small device.  Ideally, a measuring apparatus should be well-modeled classically, to reduce noise, however, that may not be compatible with the previous requirement.  Therefore, it is important to consider how a moderately sized system decoheres when being used in this way.  This work is directly applicable to the magnetic resonance force microscopy (MRFM) proposal of~\cite{RBMC04, SGBRZHY95}.  In that scheme a magnet on the scale of tens of nanometers in length is placed on tip of a nanomechanical resonator, which at very low temperatures may allow single-spin
measurements in a solid state quantum computer.  In this regime the assumption of classical behavior of the device leads to a poor approximation.  Similar problems would occur in situations in which Hamiltonians must be applied to specific subsystems without disturbing neighboring ones.
%-----

We now formally define the scenario that we are considering.  Let $\rho^{(0)}$ be the initial state of the quantum reference frame.  Let $\sigma^{\otimes n}$ be the state of a reservoir composed of $n$ ordered subsystems on which we will sequentially perform operations using the same quantum reference frame. Using the first subsystem, we apply a quantum operation $\chi$ on the joint system $\rho^{(0)} \otimes \sigma_1$.  We require that the map $\chi$ is \emph{rotationally invariant}, meaning that it is independent of any specific direction in space. Formally, a map $\chi$ is rotationally invariant if and only if $\chi[R(\Omega)(\cdot)R(\Omega)^\dag] = R(\Omega) \chi(\cdot)R(\Omega)^\dag$ for any rotation $\Omega \in SO(3)$ (or more generally $SU(2)$ if half-integral spins are considered) of the joint system.  (Here, $R$ is the unitary representation of the rotation group $SO(3)$ on the joint system.)  Subsequent to the operation, the reduced state of quantum reference frame is $\rho^{(1)}=\Tr_{r_1}[\chi(\rho^{(0)} \otimes \sigma_1)]$, where $\Tr_{r_i}$ denotes the partial trace over the $i^{\rm th}$ subsystem of the reservoir. The first subsystem of the reservoir is then discarded, and the same operation is performed on the second subsystem, but using the updated quantum reference frame $\rho^{(1)}$.  We repeat these steps with the following subsystems of the reservoir, always using the updated quantum reference frame. The state of the quantum reference frame after the $i^{\rm th}$ step is recursively $\rho^{(i+1)}=\Tr_{r_i}[\chi(\rho^{(i)} \otimes \sigma_i)]$. By discarding the previous subsystems of the reservoir at every step, we assume that the reservoir cannot be used to increase our knowledge about the quantum reference frame.

We wish to define a measure of the quality of the quantum reference frame, and to study how it decreases with the number of times the reference frame is used. Defining the quality of a quantum reference can be quite arbitrary and depends on the interests of the experimentalist and on the task at hand.  In this paper, we use the term \emph{quality function} for any function $F$ that is meant to quantify the ability of the reference frame to perform a particular task.  We will set general conditions on such functions, and analyze how such functions evolve with the number of uses of the quantum reference frame. We will also analyze the \emph{longevity} of the quantum reference frame, that is, the number of times it can be used to perform a certain task before its quality falls below a certain threshold.

As in~\cite{brst, py}, a quantum directional reference frame could be used to measure whether a series of spin-$\frac{1}{2}$ particles are parallel or anti-parallel to some classical arrow.  (See Figure~\ref{directionsfig}.)  However, there are a myriad of other possible cases, including measuring the angular momentum of a series of spin-$j'$ systems with $j'>\frac{1}{2}$, or using the quantum reference spin to indicate an axis about which a unitary rotation operation is performed.

\begin{figure}[h!]
\begin{center}
\includegraphics[width=3.5in]{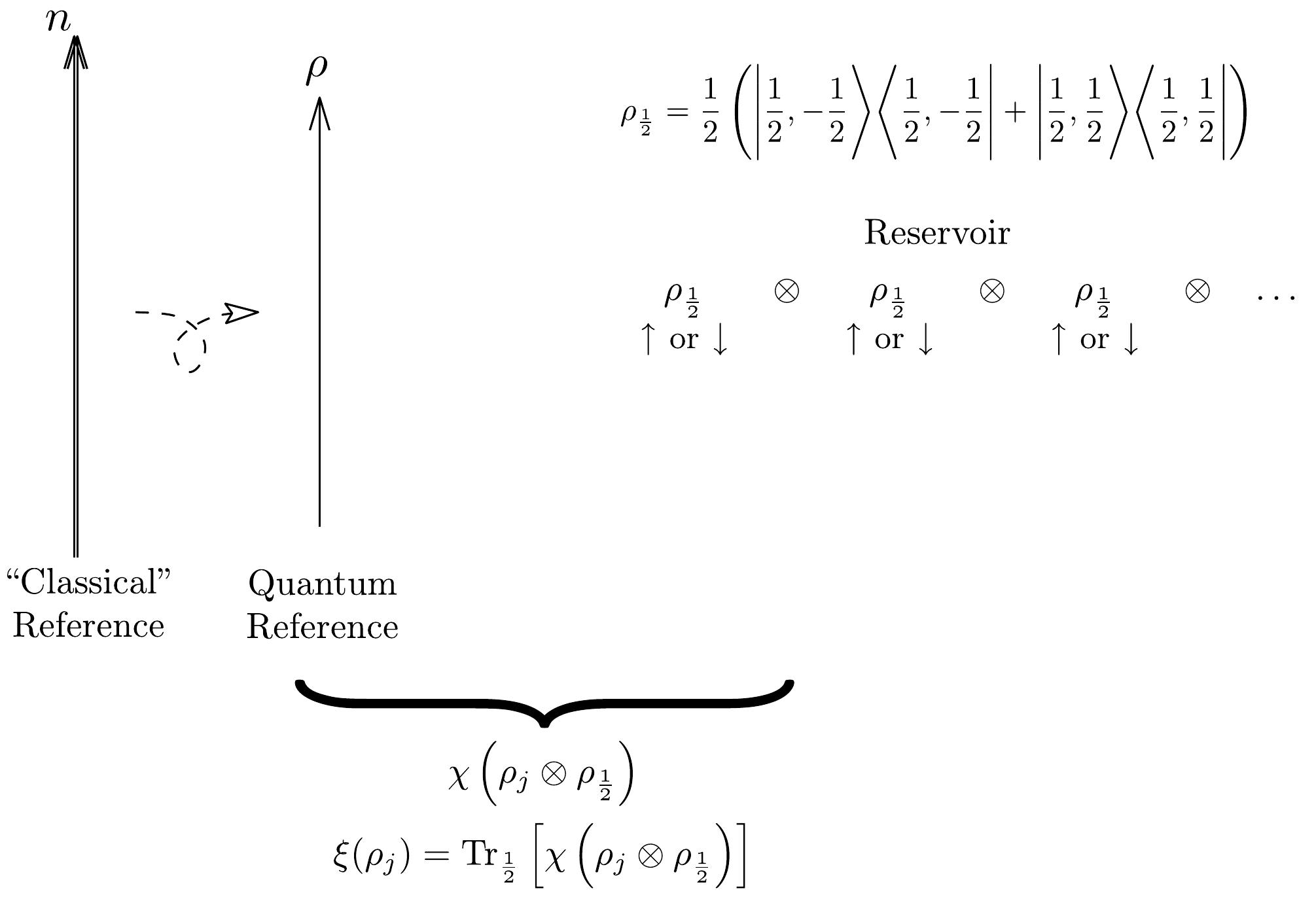}
\caption{The scenario studied by Bartlett \emph{et al.}~\cite{brst}.  The quantum reference spin is used to measure the direction of a series of spin-$\half$ particles in the completely mixed state by means of a projection onto the $J-\half$ or $J+\half$ subspaces.}
\label{directionsfig}
\end{center}
\end{figure}

In this paper, we generalize the results of Refs.~\cite{brst, py} in three different ways.  First, we consider performing operations on a reservoir of \emph{qudits} (quantum systems with $d$-dimensional Hilbert spaces) instead of restricting to the case of qubits ($d=2$) as in~\cite{brst,py}.  Second, we consider all possible rotationally-invariant quantum operations on the reference frame and the system, while the analyses in~\cite{brst,py} were restricted to very particular kind of interaction (specifically, interactions describing measurements).  Note that this joint quantum operation could be a measurement, a conditional unitary, or any other completely-positive map, provided that it is rotationally invariant.  Third, we generalize the concept of the quality function, and place general conditions on its form and evolution. In~\cite{brst,py}, only a particular kind of quality function was considered -- one which was based on the average probability of a correct measurement given a known spin direction -- while now we are interested in any possible quality function.

The paper is organized as follows. In Section~\ref{sec:momentsII}, we introduce a set of parameters for describing a quantum reference frame, which we call {\it moments}, and demonstrate that any quality function must depend only on these moments.  We then present the key result, Theorem~\ref{th1}, which provides a classification of all rotationally-invariant maps.  We give some recursive equations describing the evolution of the moments (Theorem~\ref{th2}).  In Section~\ref{sec:paramsandscaling}, we discuss the longevity of a quantum reference frame for a general quality function under the repetitive application of a rotationally-invariant map, by analyzing the dependance of the evolution of its moments. Section~\ref{sec:examples} is dedicated to a pair of examples intended to illustrate the power of the techniques developed in the earlier sections of the paper.  We explore the specific cases of measuring spin-$1$ particles relative to a quantum reference frame, as well as the use of a reference frame to perform Pauli operations on qubits.  We conclude with some general comments in Section~\ref{sec:conclusion}.

\section{Moments and Fidelity Functions}
\label{sec:momentsII}

A classical reference frame is a local convention for describing the state of a system. A quantum reference frame is a quantum system that carries information about a classical reference frame.  We define a quantum directional reference system as one which is correlated with a classical direction in space represented by the vector $\hat{n}$.  The quantum system that indicates the direction of our reference axis is taken to be a spin-$j$ system, which has dimension $d=2j+1$.  As in~\cite{py}, we consider an arbitrary initial state $\rho^{(0)}_j(\hat{n})$ for the spin-$j$ quantum directional reference frame.  (That is, the quantum reference frame transforms under rotations according to the spin-$j$ representation.)  If we additionally assume that the initial quantum reference frame depends {\it only} on the vector $\hat{n}$, we deduce a symmetry condition: if $R_j(\Omega)$ is the spin-$j$ unitary representation of the rotation $\Omega \in SU(2)$ that transforms $\hat{n}$ to $\hat{n}'$, then $R_j(\Omega) \rho_j^{(0)}(\hat{n})R_j(\Omega)^{-1}=\rho_j^{(0)}(\hat{n}')$.  That is, the set of possible initial states $\{\rho_j^{(0)}(\hat{n}) \}$ is in one-to-one correspondence with the set of directions $\hat{n}$.

Because a vector $\hat{n}$ has an invariance group -- the group of rotations about this axis -- the symmetry condition implies that the state of the quantum directional reference frame is also invariant under rotations about the $\hat{n}$-axis.  This in turn ensures that $\rho_j^{(0)}(\hat{n})$ commutes with the angular momentum operator $J_{\hat{n}}$ in the $\hat{n}$-direction, and thus $\rho_j^{(0)}(\hat{n})$ is diagonal in the basis of eigenstates of $J_{\hat{n}}$.

We now consider how the state of the quantum reference frame is updated through its use in repeated rotationally-invariant operations $\chi$ which act on both the frame and the reservoir. We assume as in~\cite{brst} that the reservoir initiates in a state that is invariant under rotations. (The completely mixed state is one example of such a state, but there exist other non-trivial states).\footnote{Note that, in~\cite{py}, they considered the case where the qubits could also be partially or fully polarized, implying that there was some initial correlation between the systems to be measured and the quantum reference frame. Generalizing their results to maps $\chi$ that do not preserve angular momentum and to reservoirs that are not polarized in the direction of the quantum reference or only composed of spin-$\half$ particles is a challenging problem.  We hope that the framework we provide here will be useful in extending our results to this general situation.} The assumptions that the initial state of the reservoir and the joint quantum operation are spatially invariant imply that the map $\xi$ used to update the state of the quantum directional reference frame, $\rho_j^{(i+1)} = \xi(\rho_j^{(i)})=\Tr_{r_i}[\chi(\rho_j^{(i)} \otimes \sigma_i)]$, is also rotationally invariant.  We denote the restricted map $\xi$ on the quantum directional reference the {\it disturbance} map.  It describes the back-action on the reference which occurs as a result of the interaction.

Let $R_{\hat{n}}(\theta)$ be the rotation by an angle $\theta$ around the vector $\hat{n}$.  Suppose that a rotationally-invariant map $\xi$ is applied to density matrix $\rho_j$ that is diagonal in a basis given by the eigenvectors of $J_{\hat{n}}$.  Because
\begin{equation}
  R_{\hat{n}}(\theta)\xi(\rho_j)R^{-1}_{\hat{n}}(\theta)=\xi(R_{\hat{n}}(\theta)\rho_j R^{-1}_{\hat{n}}(\theta))=\xi(\rho_j),
\end{equation}
then $\xi(\rho_j)$ is also diagonal in the basis given by the eigenvectors of $J_{\hat{n}}$.  As a result, the evolution of the quantum reference frame under the repeated application of the rotationally-invariant map $\xi$ can be described by $2j+1$ equations, one for each of the eigenvalues of $\rho_j$.  A different set of parameters that are equivalent to the eigenvalues is the set of \emph{moments} given by
\begin{equation}
  \{\Tr[\rho_j J_{\hat{n}}^\ell]\  |\  1 \leqslant \ell \leqslant 2j  \}.
  \label{eq:moments}
\end{equation}
Note that only $2j$ moments are necessary because the sum of the eigenvalues must be one.  The use of moments instead of eigenvalues will greatly simplify the analysis of the evolution of the quantum reference frame.

The main motivation for studying the moments of the quantum reference frame is that any quality function $F$ will depends only on these moments.  Consequently, the behavior of the different moments will determine the behavior of the different quality functions.  To see why $F$ depends only on the moments given by equation~\ref{eq:moments}, we first remark that any reasonable form of $F$ depends only on the state of the quantum reference, but that it must respect the relation $F(\rho_j)=F(R_j(\Omega) \rho_j R_j(\Omega)^{-1})$ for all rotations $\Omega \in SU(2)$ and state $\rho_j$ diagonal in the basis composed of the eigenvectors of $J_{\hat{n}}$.  In other words, the quality measure should not be biased such that it favors a quantum reference frame that is pointed in any particular direction relative to some external frame.  All directions must be equally valid.  Therefore, $F$ does not depend on the direction of $\hat{n}$, but only on the eigenvalues or the moments of $\rho_j$.  Note that the set of moments with respect to the direction $\hat{n}$ can be written as a function of the moments with respect to any other direction --- this is simply a change of basis.

Recall that $\rho_j^{(k)}$ as the state of the quantum reference frame after the $k^{th}$ joint operation with a subsystem of the reservoir, and $\rho_j^{(0)}$ as the initial state of the quantum reference frame. By our previous arguments, for a given state $\sigma^{\otimes n}$ of the reservoir and joint quantum operation $\chi$, the moments of $\rho_j^{(k)}$ must be explicit linear functions of the moments of $\rho_j^{(k-1)}$.  These results lead us to the general recursion relation,
\begin{equation}
  \Tr[\rho_j^{(k)}J_{\hat{n}}^\ell] = \sum_{i=0}^{2j} A_{i}^{(\ell)} \Tr[\rho_j^{(k-1)} J_{\hat{n}}^{i}],
  \label{eq:rec}
\end{equation}
where $A_i^{(\ell)}$ are real coefficients.

We now present a theorem which provides a classification of the different rotationally-invariant maps in terms of Lie algebra generators. %added: "in terms of lie algebra representations"
The important question of how to obtain noiseless subsystems for quantum channels constructed from generators of a Lie algebra has been first studied in \cite{LCW98}, and various proprieties of Lie algebra channels were analyzed in \cite{Ritter05}. In contrast to those works, we restrict our attention to $SU(2)$. %added: "The important question of how to obtain noiseless subsystems for quantum channel constructed from generators of a lie algebra as been first studied in \cite{LCW98}, and various other proprieties of Lie algebra channels were analyzed in \cite{Ritter05}."
 First, some notation.  Let $J_k$ for $k=x,y,z$ be the angular momentum operators for some arbitrary Cartesian frame.  On a spin-$j$ system, define the map
\begin{equation}
  \label{eq:zetamap}
  \zeta(\rho_j)= \frac{1}{j(j+1)} \sum_{k\in \{x,y,z\}} J_k \rho_j J_k,
\end{equation}
for $\rho_j$ a density matrix of the spin-$j$ system.  Let $\zeta ^{\circ n}$ denote the $n$-fold composition of $\zeta$ for $n>0$, and define $\zeta ^{\circ 0}(\rho_j)= \rho_j$.

\begin{theorem}
  Any map $\xi$ which is invariant with respect to a spin-$j$ irreducible representation of $SU(2)$ has the form
  \begin{equation}
    \xi(\rho_j) = \sum_{n=0}^{2j} q_n \zeta^{\circ n}(\rho_j),
  \end{equation}
  where $\{q_n\}$ are real coefficients.
 \label{th1}
\end{theorem}
The proof is provided in Appendix~\ref{app:proof}.

Theorem~\ref{th1} allows us an immediate simplification by restricting the number of coefficients $A_i^{(\ell)}$ of equation~(\ref{eq:rec}) that are required to calculate the evolution of the reference state.  The following theorem limits the number of
coefficients ($A_{i}^{(\ell)}$) which determine the evolution of the $J_{\hat{n}}$-moments of a spin-$j$ under the action of a general invariant channel, and can be very useful to study the behavior of moments with low degree:

\begin{theorem}
If $\ell$ is even, then
 \begin{equation}
  {\rm Tr}[\rho_j^{(k)}J_{\hat{n}}^\ell] = \sum_{i=0}^{\ell/2} A_{2i}^{(\ell)} {\rm Tr}[\rho_j^{(k-1)} J_{\hat{n}}^{2i}]
  \label{eq:even}
\end{equation}
and if $\ell$ is odd, then
\begin{equation}
  {\rm Tr}[\rho_j^{(k)}J_{\hat{n}}^\ell] = \sum_{i=1}^{(\ell+1)/2} A_{2i -1}^{(\ell)} {\rm Tr}[\rho_j^{(k-1)} J_{\hat{n}}^{2k - 1}].
  \label{eq:odd}
\end{equation}
\label{th2}
\end{theorem}

\begin{proof}
We first define the $z$-axis so that it is parallel to the vector $\hat{n}$ which describes our classical directional reference frame. Consider any rotationally-invariant map $\xi$. By Theorem~\ref{th1}, we can write $\xi(\rho_j)=\sum_{k=0}^{2j} q_k \zeta ^{\circ k}(\rho_j)$ where the coefficients $q_k$ are real numbers and $\zeta$ is given by equation~(\ref{eq:zetamap}). Therefore,
\begin{equation}
  \Tr[\xi(\rho_j) J_{z}^\ell]= \Tr[\rho_j \xi(J_z^\ell)] = \sum_{n=0}^{2j} q_n \Tr[\rho_j \zeta ^{\circ n}(J_z^\ell)].
\end{equation}
To prove our theorem, it is sufficient to show that $\zeta^{\circ n}(J_z^\ell)$ is a polynomial in $J_z$ of degree $\ell$, but where all the powers have the same  parity as $\ell$. Define $\lambda=j(j+1)$ and define the function
\begin{equation}
  G(\ell) \equiv \frac{i}{\lambda}(J_xJ_z^{\ell-1}J_y -J_yJ_z^{\ell-1}J_x).
\end{equation}
Using the commutation relations $[J_a, J_b] = i\epsilon_{a,b,c} J_c$, we can show that
\begin{equation}
  \zeta(J_z^\ell)=\zeta(J_z^{\ell-1})J_z+G(\ell),
  \label{eq:zetaex}
\end{equation}
and
\begin{equation}
  G(\ell) = G(\ell-1) J_z+\zeta(J_z^{\ell-2})-\frac{J_z^{\ell}}{\lambda}.
\end{equation}
By induction, it is easy to prove using those relations that $\zeta(J_z^\ell)$ is a polynomial in $J_z$ of degree $\ell$ where all the powers have the same parity as $\ell$. Using the fact that if $\zeta^{\circ n}(J_z^{\ell})=\sum_{s=0}^{\ell} b _s J^{s}_z$ for some complex $b_s$, then $\zeta^{\circ n+1}(J_z^{\ell})=\zeta(\sum_{s} b _s J^{s}_z)=\sum_{s} b _s \zeta(J^{s}_z)$ and we can prove by induction that $\zeta^{\circ n}(J_z^\ell)$ is a polynomial in $J_z$ of degree $\ell$ where all the powers have same the parity as $\ell$. This concludes our proof.
\end{proof}

\section{Longevity of a Quantum Reference Frame}
\label{sec:paramsandscaling}

Suppose we have a microscopic device for performing an operation or measurement on a quantum spin using a quantum reference frame. Ultimately, after many uses, the quantum reference frame will need to be reinitialized. However, we wish to make as many uses of it as possible before performing this reinitialization, without allowing the accuracy to fall below some allowed threshold. To define this accuracy we will pick some quality function that is suited to our particular purpose. Recall that any quality function $F$ will depend only on the moments of the $J_{\hat{n}}$ operator.

As in~\cite{brst, py}, we are interested in the scaling, with respect to Hilbert space dimension, of how many times a quantum reference frame can be used before the value of its quality function falls below a certain threshold.  We refer to this property as the \emph{longevity} of a quantum reference frame.  There are three important features of our analysis to highlight.  First, whereas the work of~\cite{brst, py} restricts attention to one particular quality function, we will investigate the behavior of arbitrary functions satisfying the invariance relation described in Sec.~\ref{sec:momentsII}.  Because we consider any such function, the longevity of the reference frame can be arbitrary.  However, we are able to prove some general statements about the decay of the moments, and relationships amongst them, and these results can then be used to infer the behavior of any particular quality function.  Second, rather than considering only a particular state of the reference frame (in~\cite{brst, py}, they considered the state $\rho_j^{(0)}= \ket{j,j}_{\hat{n}}\bra{j,j}$ because it was optimal for the task at hand), we consider an arbitrary state.  As such, the initial state $\rho_j^{(0)}$ of the quantum reference frame, as specified by its moments, can depend on $j$ in a quite arbitrary way.  Finally, we consider arbitrary rotationally-invariant joint quantum operations $\chi$.  In~\cite{brst, py}, $\chi$ was chosen to describe a particular measurement that was optimal for some task.  The resulting disturbance map that they considered is
\begin{equation}
  \xi(\rho_j)=(\frac{1}{2}+\frac{1}{2(2j+1)^2})\rho_j + \frac{2j(j+1)}{(2j+1)^2} \zeta(\rho_j) \,,
\end{equation}
where $\zeta(\rho_j)$ is given by equation~(\ref{eq:zetamap}). This
covariant map was also analyzed in~\cite{Ritter05}. %add "Such covariant map was analyzed in~\cite{Ritter05}"
 In~\cite{brst}, it was shown that the number of times a reference frame can be used before its quality function (which in their case depends only of the first moment) falls below a certain threshold value scales by a factor $j^2$.

In the following, we prove a general theorem about the longevity of quantum reference frames.  We start with Theorem~\ref{th1}, which states that any rotationally-invariant disturbance map can be written as $\xi=\sum_{n=0}^{2j} q_n \zeta ^{\circ n}$.  However, as we now show, if the range and values of the coefficients $q_n$ can depend on $j$ arbitrarily, it is impossible to make any general statements about the longevity.  We identify some natural assumptions for the coefficients, which then allow us to present a general theorem.

First, in general, the number of parameters $q_n$ describing a map increases with $j$; i.e., there are $2j+1$ parameters which can be significant to the evolution.  As an example, consider the complete depolarization map $\chi(\rho_j \otimes \sigma)= \frac{1}{2j+1}I_{2j+1} \otimes \frac{1}{2}I_{1/2}$.  It is straightforward to show that, for this map, all parameters $q_n$ for $0\leq n \leq 2j$ are nonzero.  However, inspired by an argument in \cite{py}, consider a rotationally-invariant map $\chi(\rho_j \otimes \sigma)$ that conserves angular momentum in the $\hat{n}$ direction, meaning that $\chi(J_{\hat{n}}) = J_{\hat{n}}$ where $J_{\hat{n}}$ is the \emph{total} angular momentum operator in the $\hat{n}$ direction.  The change of angular momentum of the quantum reference frame caused by such a disturbance map $\xi$ cannot be higher than the change in angular momentum of the subsystem $\sigma$, which is in turn bounded by the subsystem's dimension $d$.  Given that the map $\zeta(\cdot)=\frac{1}{\lambda}\sum_{k=x,y,z} J_{k} (\cdot) J_{k}$ cannot lower or increase the angular momentum in any direction by more then one unit, the disturbance map's coefficients therefore satisfy $q_n=0$ for all $n\geqslant d$, where $d$ is the dimension of each subsystem $\sigma$ of the reservoir.  Thus, if a rotationally-invariant map conserves angular momentum, then we can bound the number of non-zero coefficients $q_n$ in a way that is independent of $j$.
 
Second, if the parameters $q_n$ can depend on $j$ arbitrarily, then we would not be able to conclude anything about the longevity of the quantum reference frame. To understand this, consider the example
\begin{equation}
  \chi(\rho_j \otimes \sigma)= \alpha_j \rho_j \otimes \sigma + (1- \alpha_j)\sum_{k=|j-k|}^{j+k}\Pi_{j+k}(\rho_j \otimes \sigma)\Pi_{j+k},
\end{equation}
where $\Pi_{j+k}$ is a projection into the subspace of total angular momentum $j+k$. This map conserves the total angular momentum of the joint system. The dependence of $\alpha_j$ on $j$ will determine directly the rate of the decay of the moments in function of $j$.  Because the dependence of $\alpha_j$ can be arbitrary, then the rate of the decay of the moments can also be an arbitrary function of $j$.  We conclude that, in order to make a statement about the longevity of the quantum reference frame, we need to assume more than a bound on the number of the $q_n$ parameters or that the rotationally-invariant  map $\chi$ conserves angular momentum in the direction of $\hat{n}$.  For the following theorem, we assume that each coefficient $q_n$ converges to a constant when $j \rightarrow \infty$. We also need to make an assumption about the rate of convergence: suppose that each $q_n$ can be written as a quotient of two polynomial in $j$, such that the degree of the denominator is at least the degree of the numerator (i.e., $ q_n\leqslant O(1)$).  Finally, we assume that the state $\rho^{(0)}_j$ has the propriety that $\Tr[\rho^{(0)}_j J_{\hat{n}}^{\ell}]=O(j^{\ell})$.  A sufficient condition for this to be true is that there exists a  $\beta>0$ independent of $j$ such that $_{\hat{n}}\bra{j,j}\rho^{(0)}_j \ket{j,j}_{\hat{n}}> \beta$.  (We note that these assumptions are motivated by examples, which we explore in the following section.)

\begin{theorem}[Longevity]
Consider a quantum reference frame, realised as a spin-$j$ system with initial state $\rho^{(0)}_j$, which is used for performing a rotationally-invariant joint operation $\chi$.  If this operation induces a disturbance map $\xi=\sum_{n=0}^{2j} q_n(j) \zeta ^{\circ n}$ that satisfies the following assumptions:
\begin{enumerate}
\item there exists some $n_{max}$ such that $q_n=0$ for all $n\geqslant n_{max}$\,,
\item $ q_n\leqslant O(1)$\,,
\item $\Tr[\rho^{(0)}_j J_{\hat{n}}^{\ell}]=O(j^{\ell})$\,,
\end{enumerate}
then the number of times that such a quantum reference frame can be used before its $\ell^{th}$ moment falls below a certain threshold value scales as $j^2$.
\label{th:long}
\end{theorem}

\begin{proof}
Noting that the $J_k$ operators are self-adjoint and using the cyclic propriety of the trace, the map $\zeta$ has the property that
\begin{equation}
  \Tr[\zeta(\rho) J_{\hat{n}}^\ell]  = \Tr[\rho \, \zeta(J_{\hat{n}}^\ell)].
\end{equation}
Therefore the moments of angular momentum in the $\hat{n}$ direction after the map $\xi_j$ has been applied to the reference state $\rho_j^{(k)}$ can be expressed as
\begin{equation}
  \Tr[\rho_j^{(k+1)} J_{\hat{n}}^\ell]=  \sum_{n=0}^{2j}q_n \Tr[\rho_j^{(k)} \ \zeta^{\circ n}(J_{\hat{n}}^\ell)]
\end{equation}
Using the commutation relations, the factor $\zeta^{\circ n}(J_{\hat{n}}^\ell)$ can be expanded as a polynomial in $J_{\hat{n}}$ of degree $\ell$. Define $\lambda=j(j+1)$. Using a proof by induction on equation~(\ref{eq:zetaex}) and observing that $\sum_{k\in\{x,y,z\}} J_k J_{\hat{n}} J_k = (\lambda - 1) J_{\hat{n}}$, it is easy to show that the coefficient of the leading term will be $A_{\ell}^{(\ell)}=\sum_{n=0}^{n_{max}}[ q_n(1 - O(\frac{1}{\lambda}))]=1-O\bigl(\frac{1}{\lambda}\bigr)$ where we used the normalization condition $\sum q_n =1$, the assumption that each $q_n$ is $O(1)$ and the fact that $q_n=0$ for $n>2j$. The coefficients of the non-vanishing lower terms will be $A_{i}^{(\ell)}=O(1)$ for $i<\ell$.

This reasoning about the constants $A_k^{(i)}$ is sufficient to characterize the rate of change of the moments with repeat application of the map. To demonstrate this, let $\ell=1$. We want to find the minimum value of $t$ such that $\tr[\rho_j^{(t)} J_{\hat{n}}]<c$, for some constant $c$.  Using equation~(\ref{eq:odd}), we need to solve
\begin{equation}
  c = \Tr[\rho_j^{(0)} J_{\hat{n}} ] \left(1-O\left(\frac{1}{\lambda}\right)\right)^{t_c}.
\end{equation}
Therefore,
\begin{equation}
  t_c= O\left(j^2\right).
\end{equation}
To generalize to higher moment,  observe that $\Tr[\rho^{(0)}_j J_{\hat{n}}^{i}]\leqslant O(j^{i})$ for $i<\ell$. Recall that we assumed that $\Tr[\rho^{(0)}_j J_{\hat{n}}^{\ell}]=O(j^{\ell})$. Using equation~(\ref{eq:odd}), the facts that $A_{\ell}^{(\ell)}=1-O\bigl(\frac{1}{\lambda}\bigr)$ and  $A_{i}^{(\ell)}=O(1)$ for $i<\ell$, we can extend our result to higher odd moments by using strong induction. In other word, we can show that the minimum value of $t$ such that $\tr[\rho_j^{(t)} J_{\hat{n}}^{\ell}]<c$ for proper chosen initial state $\rho^{(0)}_j$  is $O(j^2)$. For even moments, a similar approach using equation~(\ref{eq:even}) can be used to obtain an identical conclusion.
\end{proof}

Note that in the case were the assumptions of Theorem \ref{th:long} fail, but there is a specific dependance on $j$ of the map $\chi$ and the state $\rho_j^{(0)}$, equations~(\ref{eq:even}) and~(\ref{eq:odd}) may still be useful tools to study the longevity on quantum reference frame.

\section{Examples}
\label{sec:examples}

\subsection{Measuring Spin-1 Systems}

As an example, consider the measurement of spin systems relative to a directional reference frame.  This problem was investigated in~\cite{brst,py} for the case of spin-1/2 systems.  Specifically, suppose that the reservoir consists of spin-$s$ systems which are initially in the completely mixed state.  The quantum reference direction is a spin-$j$ system with $j > s$, aligned with some classical direction $\hat{n}$.  The goal is to measure each spin to determine its component $\mu$ along the vector $\hat{n}$ by performing a joint measurement on both the system from the reservoir and the reference system.   The optimal rotationally-invariant joint operation for this task~\cite{brs04} is a POVM given by the projectors $\{\Pi_{j+\mu} | \mu \in \{-s,\ldots,s\} \}$ where $\Pi_{j+\mu}$ corresponds to a projection onto the subspace where the total angular momentum of the reference spin with the measured spin is $j+\mu$.  Define the fidelity $F_{s}$ as the probability that, when the reference is used to measure a system in a \emph{known} state $|s,\mu\rangle_{\hat{n}}$ by means of the above POVM on the joint system the result it returns is correct, assuming that all values of $\mu$ are equiprobable.

Consider briefly the case of $s=1/2$, as in~\cite{brst,py}.  The fidelity function can be expressed in terms of the first moment of the reference frame \emph{after} being used $k$ times, as
\begin{equation}
  F_{\half}^{(k)}=\tr\left[\frac{1}{2}\sum_{\mu\in \{-\half, \half\}} \Pi_{j+\mu} (\rho_j^{(k)} \otimes \ket{\mu}\bra{\mu})\right]= \frac{1}{2} + \frac{1}{2j+1} \Tr[\rho_j^{(k)}J_z]\,.
\end{equation} %got ride of  \Tr(\xi^{\circ k}(\rho_j^{(0)})), got ride of \xi(\cdot) and add ^{(k)}
Using Theorem~\ref{th2} and a simple calculation to evaluate $A^{(1)}_1$, as defined in equation~(\ref{eq:odd}), we find the fidelity in terms of the \emph{original} value of the first moment to be
\begin{equation}
  F_{\half}^{(k)} =\frac{1}{2} + \frac{\tr[\rho_j^{(0)}J_{z}]}{2j+1}\left(1-\frac{2}{(2j+1)^2}\right)^k.
\end{equation}%change \hat{n} to z
Note that the fidelity function in this case depends only on the first moment.

We now generalize to the case where we wish to measure the angular momentum of a spin-$1$ particle along some direction using a quantum direction reference frame.  (Larger spins can be handled by similar means.)  First, we determine an expression for the fidelity in terms of the moments:
\begin{equation}
  F_1^{(k)} = \tr\left[\frac{1}{3}\sum_\mu \Pi_{j+\mu} (\rho_j^{(k)} \otimes \ket{\mu}\bra{\mu})\right].
  \label{eq:fidelity}
\end{equation}
where $\mu \in \{-1, 0, 1\}$.  Consider the disturbance map $\xi$ as described in Section~\ref{sec:momentsII}:
\begin{equation}
  \xi(\rho_j^{(k)}) = \frac{1}{3} \sum_{\mu'',\mu',\mu} \bra{\mu''} \Pi_{j+\mu'} \ket{\mu} \ \rho_j^{(k)} \ \bra{\mu} \Pi_{j+\mu'} \ket{\mu''}.
\end{equation}
Explicit expressions for the reduced operators on the reference system $\bra{\mu''} \Pi_{j+\mu'} \ket{\mu}$ can be found from the Clebsch-Gordon coefficients.  The expression for the fidelity in terms of $j$ and the moments is
\begin{equation}
  F_1^{(k)} = \frac{1}{6} + \frac{\left[(2j+1)^2 - 2\right]}{6j(j+1)(2j+1)} \tr[\rho_j^{(k)} J_z] + \frac{1}{2j(j+1)} \tr[\rho_j^{(k)} J_z^2].
\end{equation}
Note that the fidelity in this case depends on the second moment as well as the first. Using Theorem~\ref{th2}, we have:
\begin{eqnarray}
  \tr[\xi(\rho_j^{(k)}) J_z] &=& A_1^{(1)} \tr[\rho_j^{(k)} J_z],   \\
  \tr[\xi(\rho_j^{(k)}) J_z^2] &=& A_0^{(2)} + A_2^{(2)} \tr[\rho_j^{(k)} J_z^2].
\end{eqnarray}

Substituting some particular values of $\rho_j$ (i.e., $\rho_j= \ket{j, m}_{\hat{n}} \bra{j, m}$ for $m=j,j{-}1$) and solving the system of equations, we obtain:
\begin{eqnarray}
  A_1^{(1)} &=& \frac{3j^4 + 6j^3 - j^2 - 4j + 2}{3j^2(j+1)^2} = 1 - \frac{4}{3j^2} + O\left(\frac{1}{j^3}\right),   \\
  A_0^{(2)} &=&  \frac{2\left(8j^4 + 16j^3 - 8j - 3\right)}{3j (j+1) (2j+1)^2} = \frac{4}{3} - \frac{5}{3j^2} + O\left(\frac{1}{j^3}\right),  \\
  A_2^{(2)} &=& \frac{4j^6 + 12j^5 - 3j^4 - 26j^3 + j^2 + 16j + 6}{j^2 (j+1)^2 (2j+1)^2} = 1- \frac{4}{j^2} + O\left(\frac{1}{j^3}\right).
\end{eqnarray}
(The series apply as $j\rightarrow \infty$.)  In terms of these constants, the fidelity evolves with $k$, the number of applications of the map corresponding to the measurement of the $z$ projection of the spin-1 system, as:
\begin{equation}
  F_1^{(k)} = \frac{1}{6} + \frac{\left[(2j+1)^2 - 2\right]}{6j(j+1)(2J+1)}\left(A_1^{(1)}\right)^k \tr[\rho_j^{(0)} J_z] + \frac{1}{2j(j+1)} \left( A_0^{(2)} \frac{1-\left(A_2^{(2)}\right)^k}{1-A_2^{(2)}} + \left(A_2^{(2)}\right)^k \tr[\rho_j^{(0)} J_z^2]\right).
\end{equation}

\begin{figure}[h!]
\begin{center}
\includegraphics[width=3.8in]{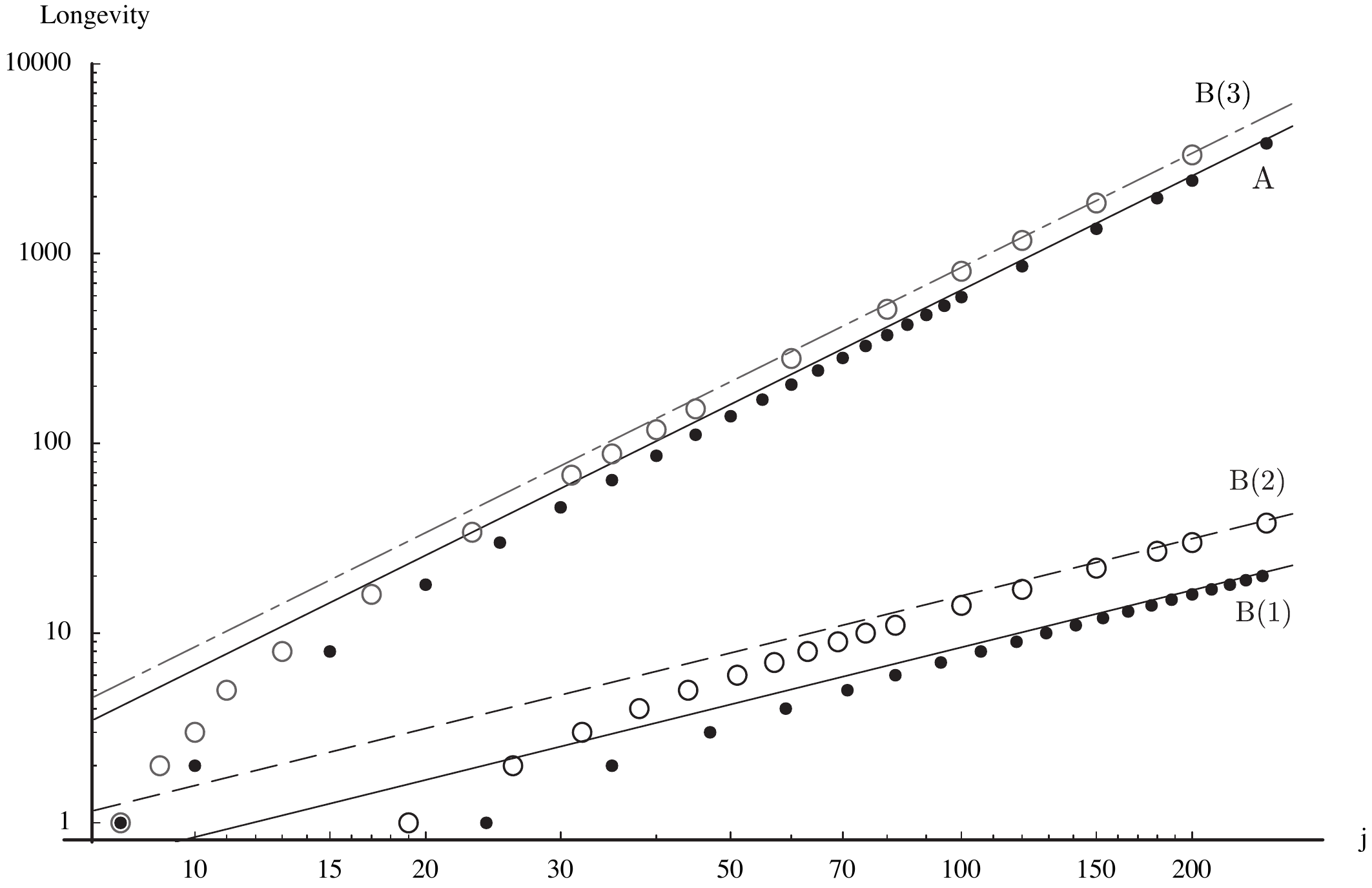}
\caption{A plot of the longevity, as defined in Section~\ref{sec:paramsandscaling}, against the reference system size $j$ for the scenario described in Example A (points converging to the solid line of gradient 2), and Example B, Method 1 (points converging to the solid line of gradient 1), Method 2 (circles converging to the dashed line of gradient 1), and Method 3 (grey circles converging to grey dot-dashed line of gradient 2).  The points are found numerically.  The gradients of the lines in A and B(3) are 2, indicating longevity scaling as $O(j^2)$, and the gradients of the lines in B(1) and B(2) are 1, indicating a scaling of $O(j)$.}
\label{fig:logplots}
\end{center}
\end{figure}

Because the assumptions given in Section~\ref{sec:paramsandscaling} are satisfied, the longevity of the quantum reference frame must scale as $O(j^2)$. This result is easily verified numerically.  See Fig.~\ref{fig:logplots}.

\subsection{Implementing a Pauli Operator on a Qubit}
\label{subsec:GateFidelity}

The quantum reference frame can also be used to define an axis for the purpose of implementing some desired unitary gate.  Having only one axis, we are restricted to the set of gates corresponding to rotations about this axis.  Suppose, for example, we want to implement a Pauli $Z$ operation on a qubit,
\begin{equation}
  Z=\left(\begin{array}{cc}
  1 & 0  \\
  0 & -1
  \end{array}\right),
\end{equation}
using a quantum reference direction $\rho_j$ in the form of a spin-$j$ system to define the $z$-axis. To do this, we must implement a rotationally-invariant operation on the combined reference and system.  Unlike the measurement considered in previous examples, in this case we wish to implement a unitary operation on the system that is conditional on the state of the quantum reference direction.  In the language of quantum information, we want to perform a unitary rotation about a direction \emph{programmed} by the state of the reference frame.  However, that because the size of the reference frame is bounded and the number of programmed operations (possible directions) is continuous, then the impossibility of a ``programmable quantum gate''~\cite{NielsenChuang97} ensures that an ideal such conditional unitary cannot exist.  It is possible to approximate it, as we now show.

\subsubsection{Method 1}

Consider the following operation, to be performed on the joint state of a quantum reference frame $\rho_j$ and a single qubit $\sigma$.  First, define an covariant measurement on the spin-$j$ reference frame described by the POVM
\begin{equation}
  \{ \Lambda(\Omega) = (2j+1)R(\Omega)|e\rangle\langle e|R(\Omega)^{\dagger} \,, \ \Omega \in SU(2) \} \,,
\end{equation}
where $|e\rangle$ is a normalized state of a spin-$j$ system.  The measurement effects of this POVM satisfy the normalization condition
\begin{equation}
  \int_\Omega d\mu_\Omega \Lambda(\Omega) = I_j \,,
\end{equation}
where $I_j$ is the identity on the spin-$j$ Hilbert space and $d\mu_{\Omega}$ represents the Haar measure over $SU(2)$. %add " and $\mu_{\Omega}$ represents the Haar measure over $SU(2)$"
 This measurement is performed on the quantum reference frame state $\rho_j$, and then conditional upon obtaining the outcome $\Omega$, the operation $Z(\Omega) = R(\Omega) Z R^{-1}(\Omega)$ is applied to the system.  If we subsequently discard the information about the measurement result $\Lambda$, then the effective joint map
\begin{equation}
  \chi(\rho_j \otimes \sigma)= \int_{\Omega} d\mu_{\Omega} \sqrt{\Lambda(\Omega)}\otimes Z(\Omega)(\rho_j \otimes \sigma)  \sqrt{\Lambda(\Omega)} \otimes Z(\Omega)^{\dagger}\,,
\end{equation}
is clearly invariant.  (Although the covariant measurement makes use of an external reference frame, the resulting map $\chi$ is independent of this choice of frame.) The net operation on the system $\sigma$ is given by the map
\begin{equation}
  \tau (\sigma) = \int_\Omega d\mu_\Omega \tr_j \left[ \Lambda(\Omega) \rho_j\right] Z(\Omega)\sigma Z^{-1}(\Omega)\,.
\end{equation}
Note that this operation $\tau$ is also rotationally invariant; although the measurement and subsequent unitary appear to require an external spatial reference frame, the net operation is invariant under changes of this frame.   Also, because the operation is constructed explicitly from a POVM measurement and a unitary operation conditional on this classical result, followed by tracing out the state of the reference frame, it is necessarily completely positive and trace preserving (CPTP).  (This fact is also clear by observation:  the term $\tr_j \left[ \Lambda(\Omega) \rho_j\right]$ is a normalized probability distribution weighting possible unitaries $Z(\Omega)$, and thus the map $\tau$ is a valid unital CPTP map.)

This expression for $\tau$ explicitly gives a Kraus decomposition of this map (albeit with a continuous number of Kraus operators), $\tau(\sigma)= \int_\Omega d\mu_\Omega E(\Omega) \sigma E(\Omega)^{\dagger}$, where
\begin{equation}
  E(\Omega) = \sqrt{\tr_j[ \Lambda(\Omega) \rho_j]} Z(\Omega)\,,
\end{equation}
are Kraus operators satisfying $\int_\Omega d\mu_\Omega E(\Omega)^\dag E(\Omega) = I$.  The ability of this operation to approximate the $Z$ operation on the system is defined using the \emph{gate fidelity} \cite{MPRHorodecki99, BowdreyOiShortBanaszekJones02, Nielsen02, EmersonAlickiZyczkowski05} given by
\begin{equation}
  F_{\rm gate}(Z, \tau) \equiv \frac{\int_\Omega d\mu_\Omega \bigl|\Tr[E(\Omega)^{\dagger} Z]\bigr|^2+d}{d^2+d}
\end{equation}
where $d$ (the dimension of the system on which the gate is applied) is 2 in our case.

Suppose that quantum directional reference frame $\rho_j$ is aligned with the $z$-axis. Then $\rho_j$ is a mixture of states $\ket{j,m}_{z}$ for $-j \leqslant m \leqslant j$.  To simplify our calculation, assume that $\rho_j=\ket{j, m}_z\bra{j,m}$ (i.e., we consider only a pure state, but at the end, we can generalize for any mixed state). The state $|e\rangle$ that defines the POVM could be any state, but for simplicity we consider only the case where $|e\rangle = |j,j\rangle_z$. In this case, the rotationally-invariant operation on the combined reference and system is
\begin{equation}
  \tau (\sigma, \rho_j) = (2j+1) \int_\Omega d\mu_\Omega \bigl|{}_z \langle j,m|R(\Omega)|j,j\rangle_z  \bigr|^2 Z(\Omega)\sigma Z^{-1}(\Omega)\,.
\end{equation}
The Kraus operators are then given by
\begin{equation}
  E(\Omega) = \sqrt{(2j+1)}\bigl|{}_z \langle j,m|R(\Omega)|j, j \rangle_z\bigr| Z(\Omega)\,.
\end{equation}

We can parameterize
\begin{equation}
  R(\Omega)=e^{i\alpha} \left( \begin{array}{cc} e^{i\phi} \cos \theta & e^{i\psi} \sin \theta  \\ -e^{-i\psi} \sin \theta & e^{-i\phi} \cos \theta   \\  \end{array} \right)
\end{equation}%I reintegrate the parameter \alpha
where $0 \leq \alpha, \psi,\phi \leq 2\pi$ and $0 \leq \theta \leq \frac{\pi}{2}$. The rotation $R(\Omega)$ can be rewritten using Euler angles:
\begin{equation}
  R(\Omega)=e^{i\alpha} R_z(-\phi-\psi)R_y(-2\theta)R_z(\psi-\phi),
\end{equation}%I reintegrated the parameter \alpha
where $R_k(\beta)$ is a clockwise rotation of angle $\beta$ around the $k$-axis. It follows that
\begin{equation}
\Tr[R(\Omega)ZR^{\dagger}(\Omega)Z]= \Tr[R_y(-2\theta)ZR_y(2\theta)Z] =2\cos{2\theta}.
\end{equation}
From~\cite{rose}, we have that
\begin{equation}
  \ _z \bra{j,m}R(\Omega)\ket{j,j}_z = e^{-i\alpha-im(\phi+\psi)-ij(\phi-\psi)} \sqrt{{2j \choose j+m }} (\cos \theta)^{j+m}(-\sin \theta)^{j-m}.
\end{equation}%I reintegrated the parameter \alpha
Therefore,
\begin{align}
  |\Tr[E(\Omega)^{\dagger}Z] |^2&=(2j+1) {2j \choose j+m } (\cos \theta)^{2j+2m}(\sin \theta)^{2j-2m} \bigl|\Tr[R_y(-2\theta)ZR_y(2\theta)Z]\bigr|^2 \nonumber \\
  &=4 (2j+1){2j \choose j+m } (\cos \theta)^{2j+2m}(\sin \theta)^{2j-2m} (\cos{2\theta})^2,
\end{align}
where, in the first step, we used the cyclic propriety of the trace.  The Haar measure in these coordinates is $d\Omega (U(2))= (2\pi)^{-3} 2 \sin \theta \cos \theta d\theta d\alpha d\phi d\psi$.\footnote{We take the integral over the Haar measure for $U(2)$ instead of $SU(2)$ because this simplifies the notation, however, it should be noted that in this case the integral over $\alpha$ (the global phase) will not
alter the result, so we are free to do this.}  It follows that%I reintegrated the parameter \alpha and consider the Haar measure over U(2) (simpler to write than the Haar measure over SU(2) which required that we define the sign function) ***** & added a footnote  *****
\begin{equation}
  F_{gate}(Z, \tau) = \frac{1}{3}+ \frac{2(j+1+2m^2)}{3(j+1)(2j+3)}
\end{equation}
which is a function of the second moment of the $z$ projection of the reference frame. The result easily generalizes to the case where $\rho_j$ is a mixture of pure states of the form $\ket{j,m}_{z}$: %minor grammar fix
\begin{equation}
  F_{gate}(Z, \tau) = \frac{1}{3}+ \frac{2(j+1+2 \Tr[\rho_j J_z^2])}{3(j+1)(2j+3)}.
  \end{equation}%small changes
This example is a case where the fidelity evolves with the expectation value of the second moment, despite the fact that the reservoir is composed of spin-$\half$ particles. Also note that in this particular case the fidelity does not depend on any odd moments because the operation depends only on the axis defined by $\hat{n}$, but not on the direction along this axis.

We now consider how the quantum reference frame degrades with repeated use in performing this operation.  We assume that the qubit systems (the reservoir) on which we apply the approximate phase gate are all in the completely mixed state $ \frac{1}{2} I_{\frac{1}{2}}$.  With each application, the reference frame evolves according to the invariant map $\xi(\rho_j)= \Tr_s \chi(\rho_j \otimes \frac{1}{2}I_{\frac{1}{2}})$, where $\Tr_s$ is the partial trace of the qubit system on which we apply the approximate phase gate.
To calculate how the second moment evolves as a function of the number of times the quantum reference frame has been used to perform the approximate phase gate, we can use Theorem \ref{th2}:
\begin{equation}
  \Tr[\xi(\rho_j) J_z^2]= A_0^{(2)}+A_2^{(2)} \Tr[\rho_j J_z^2].
\end{equation}
To find the values of the coefficients $A_0^{(2)}$ and $A_2^{(2)}$, we consider two possible initial states of the quantum reference frame.  First, suppose that $\rho_j= \frac{1}{2j+1}I_j$, which yields $\xi(\rho_j)= \frac{1}{2j+1}I_j$. This evolution gives
\begin{equation}
  \frac{j(j+1)}{3}= A_0^{(2)}+A_2^{(2)}\frac{j(j+1)}{3}.
\end{equation}
Second, using $\rho_j= \ket{j, j}_z \bra{j,j}$,
we obtain a second equation:
\begin{equation}
  \frac{j(1+j(3+j+2j^2))}{(1+j)(3+2j)}= A_0^{(2)}+A_2^{(2)}j^2.
\end{equation}
Solving those equations, we obtain
\begin{equation}
  A_0^{(2)}=j-\frac{2j}{2j+3} = j - 1 + \frac{3}{2j} - \frac{9}{4j^2} + O\left(\frac{1}{j^3}\right),
  \label{eq:azerob}
\end{equation}
and
\begin{equation}
  A_2^{(2)}=1-\frac{3(2j+1)}{(2j+3)(j+1)} = 1- \frac{3}{j} + O\left(\frac{1}{j^3}\right).
  \label{eq:atwob}
\end{equation}

%got rid of upside-down question mark
From the above equations, we can deduce that the longevity of the quantum reference frame is $O(j)$ by noting that the fidelity after $k$ repetitions of the gate is %got rid of >
\begin{equation}
F^{(k)}_{gate} (Z, \tau) = \frac{1}{3} + \frac{2}{3(j+1)(2j+3)}\left(j+1+ 2 A_0^{(2)} \frac{1-\left(A_2^{(2)}\right)^k}{1-A_2^{(2)}} + 2\left(A_2^{(2)}\right)^k \tr[\rho_j^{(0)} J_z^2]\right).
\label{eq:fidb}
\end{equation}

%got rid of upside-down question mark
From the equations~(\ref{eq:azerob}) and~(\ref{eq:atwob}) we have that $A_0^{(2)}$ goes as $O(j)$ and $A_2^{(2)}$ goes as $O(1)$.  Looking at equation~(\ref{eq:fidb}) we see that fidelity must go as $O(\frac{1}{j})$ and therefore longevity goes as $O(j)$.  This is also seen in the numerical work shown in Figure~\ref{fig:logplots}. %got rid of >
 This result appears to be in contradiction with Theorem \ref{th:long}.  However, this contradiction is resolved by the fact that one of the assumptions of the theorem is not fulfilled.  Indeed, even if the map $\chi$ is rotationally invariant, it does not conserve the total angular momentum.  In particular, note that $\bra{j, m}_z \xi_j(\ket{j,j}_z \bra{j,j}) \ket{j,m}>0$ for all $m$ in the range $-j,\ldots,j$.  Because $\zeta^{\circ 2j}(\ket{j,j}_z \bra{j,j})$ is the only map of the form $\zeta^{\circ n}(\ket{j,j}_z \bra{j,j})$ for $0<n<2j+1$ that has support in the state $\ket{j,-j}_z \bra{j,-j}$, then $q_{2j}>0$. Therefore, there is no bound $n_{max}$ independent of $j$ such $q_n =0$ for all $n>n_{max}$.

\subsubsection{Method 2}
\label{subsec:MoreGateFidelity}

We now investigate an alternate method for approximating a Pauli $Z$ operation using a quantum reference frame.  Consider the filtering operation \cite{G96}
\begin{equation}
  \Gamma = (2j+1) \int_{\Omega} d\mu_\Omega R_j(\Omega) \otimes R_{1/2}(\Omega) \Bigl[ |j,j\rangle _z \langle j,j| \otimes Z \Bigr] R_j(\Omega)^{-1} \otimes R_{1/2}(\Omega)^{-1}.
  \label{eq:filter}
\end{equation}
Observing that in spin notation%add _z
\begin{equation}
  Z=|\tfrac{1}{2},\tfrac{1}{2}\rangle _z \langle \tfrac{1}{2},\tfrac{1}{2}| - |\tfrac{1}{2},-\tfrac{1}{2}\rangle_z \langle \tfrac{1}{2},-\tfrac{1}{2}|,
\end{equation}%add _z
we can replace the operator $Z$ in equation~(\ref{eq:filter}) with this expression and using the Clebsch-Gordon coefficients we can express the part of equation~(\ref{eq:filter}) that is in the square brackets as a joint system with total angular momentum having support on the $j-\frac{1}{2}$ and $j+\frac{1}{2}$ subspaces. %Because $\Gamma$ has an integral over all $\Omega$, it
 $\Gamma$ is spatially-covariant and we can use Schur's lemma to note that it will be block-diagonal in the irreps $j+\frac{1}{2}$ and $j-\frac{1}{2}$.  Thus, we can rewrite it as
\begin{align}
  \Gamma &= (2j+1) \int_{\Omega} d\mu_\Omega R_{j+1/2}(\Omega) \Bigl[ |j+\tfrac{1}{2},j+\tfrac{1}{2}\rangle_z \langle j+\tfrac{1}{2},j+\tfrac{1}{2}| - \tfrac{1}{2j+1}|j+\tfrac{1}{2},j-\tfrac{1}{2}\rangle_z \langle j+\tfrac{1}{2},j-\tfrac{1}{2}|\Bigr] R_{j+1/2}(\Omega)^{-1}  \nonumber  \\
& \qquad \qquad \qquad  -(2j+1) \int_{\Omega} d\mu_\Omega R_{j-1/2}(\Omega)  \Bigl[ \tfrac{2j}{2j+1}|j-\tfrac{1}{2},j-\tfrac{1}{2}\rangle_z \langle j-\tfrac{1}{2},j-\tfrac{1}{2}| \Bigr] R_{j-1/2}(\Omega)^{-1}   \\
   &= (2j+1) \Bigl( \tfrac{1}{2j+2} (1-\tfrac{1}{2j+1}) I_{j+1/2} - \tfrac{1}{2j} \tfrac{2j}{2j+1} I_{j-1/2} \Bigr),
\end{align}
which simplifies to
\begin{equation}
\Gamma =( \frac{2j}{2j+2} I_{j+1/2} - I_{j-1/2}).
\end{equation}
We can define
\begin{equation}
\chi_j(\rho_j \otimes \sigma)=  \Gamma (\rho_j \otimes \sigma) \Gamma^{\dagger},
\end{equation}
where $\Pi_{j+k}$ is a projection into the subspace of total angular momentum $j+k$. However, this map is not trace preserving, nor does it conserve angular momentum. As in method 1, let $\rho_j= \ket{j,m}_z\bra{j,m}$. %add "As in method 1, let $\rho_j= \ket{j,m}_z\bra{j,m}$."
 In that case, we have Kraus operators of the form
\begin{eqnarray}
E_{m'} &=& _z \bra{j, m'} \Gamma \ket{j, m} _z
\end{eqnarray}
and, using the appropriate Clebsch Gordon coefficients, we arrive at an expression for fidelity of the form
\begin{equation}
F_{gate} = \frac{1}{3} + \frac{2}{3}\left(\frac{1}{j+1}\right)^2 m^2,
%F_{gate} = \frac{1}{3} + \frac{2}{3}\left(\frac{1}{2j+1}\right)^2 \langle m^2\rangle \left(\frac{j^2}{(j+1)^2}+1+\frac{2j}{(j+1)}\right),
\end{equation}%I simplify the equation
which goes to $1$ in the limit $j\rightarrow \infty$, for the case $m=j$. For a general state $\rho_j$ (which must be diagonal in the basis $\{\ket{j,m}_z\ |\ -j \leqslant m\leqslant j \}$), we obtain
\begin{equation}
F_{gate} = \frac{1}{3} + \frac{2}{3}\left(\frac{1}{j+1}\right)^2 \Tr[\rho_j J_z^2],
\end{equation}%add the last sentence

By Theorem~\ref{th2}, even though the map is not trace preserving, we have:%add" By Theorem \ref{th2}," and got rid of "the same expressions for the evolution of the moments"
\begin{equation}
 \Tr[\xi(\rho_j) J_z^2] = A_0^{(2)}+A_2^{(2)} \Tr[\rho_j J_z^2],
 \end{equation}
and, by the same method as in the previous subsection, %change to "the previous subsection"
 we have for this map
\begin{eqnarray}
A_0^{(2)} &=& \frac{j}{(j+1)}  = 1 - \frac{1}{j} + \frac{1}{j^2} + O\left(\frac{1}{j^3}\right),  \\
A_2^{(2)} &=&  \frac{j}{(j+1)}\left(1 - \frac{3}{j(j+1)}\right) = 1- \frac{1}{j} - \frac{2}{j^2} + O\left(\frac{1}{j^3}\right),
\end{eqnarray}
where the series expansions are appropriate in the limit $j \rightarrow \infty$.

Again we find the scaling of longevity with $j$ to be $O(j)$.  (See Figure~\ref{fig:logplots}.)  However, the longevity is consistently higher than by using Method 1.

\subsubsection{Method 3}
\label{subsec:EvenMoreGateFidelity}%Martin's modifications of the whole subsection

We might ask about performing an operation similar to the previous case, but in which trace is preserved on the map.  Such a map could be realized through a coupling between the spins of the form:
\begin{equation}
H = w\left(J_{x}^{\left(\frac{1}{2}\right)} J_{x}^{(j)} + J_{y}^{\left(\frac{1}{2}\right)} J_{y}^{(j)} + J_{z}^{\left(\frac{1}{2}\right)} J_{z}^{(j)}\right) = \frac{w}{2} \mathcal{J}^2 - \frac{w}{2}\left(j(j+1)+\frac{3}{4}\right)\openone
\label{eq:spincoupling}
\end{equation}
where $w$ is some constant, the operator $J_{x}^{\left(\frac{1}{2}\right)} = X$ is the $J_x$ operator on a spin-half system and the operator $\mathcal{J}^2 =\left(j+\frac{1}{2}\right)\left(j+\frac{3}{2}\right) \Pi_{j+\frac{1}{2}} + \left(j-\frac{1}{2}\right)\left(j+\frac{1}{2}\right) \Pi_{j-\frac{1}{2}}$ is the total angular momentum operator on the joint system of the reference and the spin-half system.  Therefore, the unitary associated with this Hamiltonian is $U=e^{i\frac{w}{2} (\mathcal{J}^2 -(j^2+j+\frac{3}{4})\openone)t}$.  Choosing $t = \frac{2\pi}{w(2j+1)}$ gives the appropriate phase shift of the spin-half system relative to the spin-$j$ reference system.

It could also be phrased in the language of the previous subsection %add sub
 if --- up to a global phase --- we set%add "up to a phase"
\begin{equation}
\Gamma =U \cong (I_{j+1/2} - I_{j-1/2}).
\end{equation}

The quality function for this map is
\begin{equation}
F_{gate} = \frac{1}{3} + \frac{2}{3}\left(\frac{2}{2j+1}\right)^2  \Tr[\rho_j J_z^2],
\end{equation}%change \langle m^2\rangle to  \Tr[\rho_j J_z^2]
which, when $ \Tr[\rho_j J_z^2]= j^2$, also tends to $1$ in the limit $j\rightarrow \infty$.%change \langle m^2\rangle to  \Tr[\rho_j J_z^2]

We have for this map
\begin{eqnarray}
A_0^{(2)} &=& 1 - \frac{1}{(2j+1)^2} = 1- \frac{1}{4 j^2}  + O\left(\frac{1}{j^3}\right),   \\
A_2^{(2)} &=&  1 - \frac{12}{(2j+1)^2} = 1 - \frac{3}{j^2} + O\left(\frac{1}{j^3}\right),
\end{eqnarray}
where, again, the series expansions are appropriate in the limit $j \rightarrow \infty$.  Observe that in this case these parameters lack any terms in $\frac{1}{j}$.

Now we find the scaling of longevity with $j$ to be $O(j^2)$, as depicted in Figure~\ref{fig:logplots}.  Figure~\ref{fig:fidplots} compares the fidelities for the three methods for constant $j$. These last three examples demonstrate the importance of choosing the gate carefully to match the objectives of the given task.

\begin{figure}[h!]
\begin{center}
\includegraphics[width=3.8in]{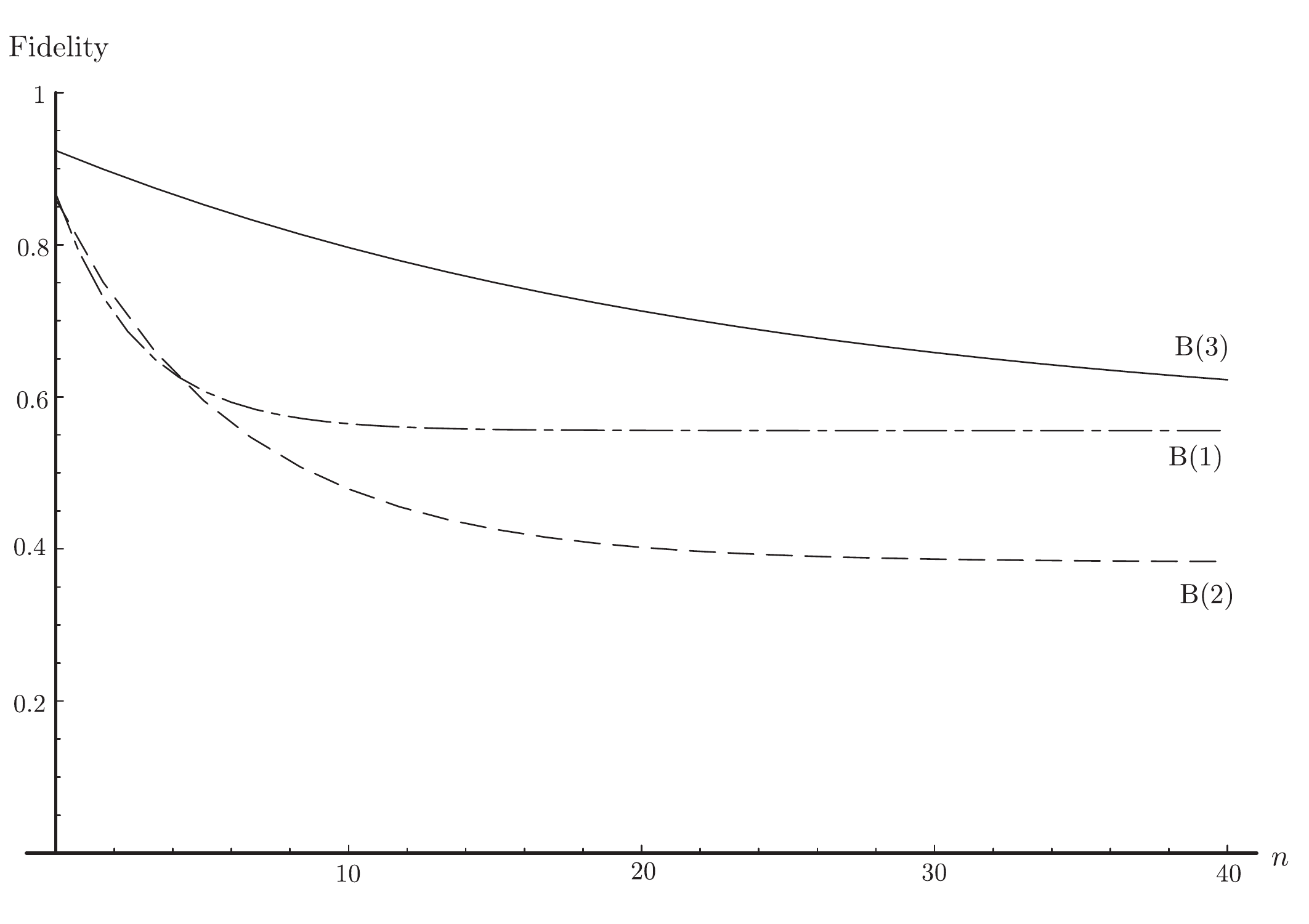}
\caption{A plot of the fidelity with number of repetitions, $n$, for $j=8$ for the three methods, B(1) (dot-dashed line), B(2) (dashed line), and B(3) (solid line).  This behavior of this value of $j$ is representative.}
\label{fig:fidplots}
\end{center}
\end{figure}

\section{Conclusion}
\label{sec:conclusion}

This study has explored the physically relevant question of how a quantum directional reference frame state evolves under the action of maps invariant with respect to rotations in space. We generalize the concept of quality of a quantum directional reference frame introduced by~\cite{brst}. We argue the quality of a reference frame must be represented by a function that depends only on the eigenvalues of the quantum reference frame or an equivalent set of parameters called the moments. We give recursive equations (Theorem~\ref{th2}) for how the moments evolve with the number of uses of the quantum reference frame. We derive sufficient some conditions (Theorem~\ref{th:long}) for the longevity of a quantum reference frame to scale by a factor proportional to square the dimension of the quantum reference frame.  Finally, we applied our results to different examples such as the use of a quantum directional reference frame to measure a spin-1 particle or  to implement an Pauli operator on a qubit. The tools that we developed can be use to compare different methods to perform some operation using a quantum reference frame as we showed in our last example.

To prove Theorem~\ref{th2}, we use Theorem~\ref{th1} which implies that maps invariant with respect to $SU(2)$ can be written in a polynomial form as a function of the Lie algebra generators. It would be interesting to investigate if a similar theorem could be applied to other Lie groups. Also, we assume in this analysis that the state of the reservoir is invariant under rotations. However, in the most general physical situation, the reservoir can be polarized. Building on the work done in~\cite{py}, it would be interesting to investigate how our theorem can be generalized to the case where the above symmetry is broken.
%added
%These techniques can be generalized straightforwardly to other groups and their corresponding reference frames, for example, using techniques from~\cite{Ritter05}, it would be possible to consider the groups $SU(n)$, however it seems this case is not physically motivated outside of $SU(2)$.
%

%----added-----

%An experiment to confirm our findings regarding the scaling of the degradation of a directional reference frame under use to perform a $Z$ gate could be performed in a very straightforward way in an NMR setting for Method 3.  A dimer molecule with a low spin species and a high spin species in a magnetic field will experience the Hamiltonian given in equation~\ref{eq:spincoupling}.  

%----added-----

\begin{acknowledgments}
We wish to acknowledge the contribution of Terry Rudolph to show that the coefficients $q_i$'s in Theorem \ref{th1} are real. We also thank Giacomo Mauro D'Ariano, Joseph Emerson, Michael Keyl, David Kribs, David Poulin and Peter Turner for helpful discussions. J.-C.B and M.L. are supported by NSERC, and L.S. by the Mike and Ophelia Lazaridis Fellowship.  S.D.B. acknowledges the support of the Australian Research Council.
\end{acknowledgments}

\appendix
\section{Proof of Theorem 1}
\label{app:proof}

Here, we provide a proof of Theorem~\ref{th1}, inspired by ideas from~\cite{KeylWerner99, Ariano04}. First, because the group $SO(3)$ of space rotations and the group $SU(2)$ of unitary transformation are locally isomorphic up to a phase, then the spatial covariance condition can be replaced by covariance with respect to $SU(2)$.

Let $\rho_j$ be a density operator on the Hilbert space of a spin-$j$ system.  Consider the map
\begin{equation}
  \zeta(\rho_j)= \frac{1}{\lambda} (J_x \rho J_x+J_y \rho J_y+J_z \rho J_z),
\end{equation}
where $\lambda=j(j+1)$, and $J_x$, $J_y$ and $J_z$ are the angular momentum operators in the $x$, $y$ and $z$ directions for some arbitrary Cartesian frame.  First, we notice that
\begin{equation}
  \zeta(I_j) = \frac{1}{\lambda}\sum_{i=1}^3 J_i^\dag J_i =\frac{1}{\lambda}J^2=I_j,
\end{equation}
where $J$ is the total angular momentum operator.

\begin{lemma}
The map $\zeta$ is covariant with respect to the spin-$j$ irreducible representation $R_j$ of $SU(2)$.%covariant instead of invariant
\end{lemma}

\begin{proof}%%%%modified considerably the argument!!!!!!
If $R_j$ is the $d=2j+1$ dimensional irreducible representation (irrep) of the group $SU(2)$ generated by $\{J_x,J_y,J_z\}$, then it can be expressed in the form
\begin{equation}
  R_j(\Omega)=e^{-i\theta_1 J_z}e^{-i\theta_2 J_y}e^{-i\theta_3 J_z},
   \label{eq:expandrot}
\end{equation}
for some real parameters $\theta_1$, $\theta_2$ and $\theta_3$ (which are strongly related to the Eulers angles). By the symmetry of $\zeta$ with respect to the $y$ and $z$-axes, if we prove that
\begin{equation}
  R_z(\theta)^{-1} \zeta(R_z(\theta) \rho_j R_z(\theta)^{-1} )R_z(\theta)=\zeta(\rho_j), \qquad \forall\ \theta \in [0,2\pi),
 \label{eq:covcon2}
\end{equation}
for any rotations around the $z$-axis only, then (\ref{eq:covcon2}) with any rotations around the $y$-axis follows immediately. From  (\ref{eq:expandrot}), we deduce that (\ref{eq:covcon2}) is satisfied for any rotations. Such $z$ rotations will map
\begin{eqnarray}
  J_x&\to& \cos{\theta}J_x+\sin{\theta}J_y\\
  J_y&\to& -\sin{\theta}J_x+\cos{\theta}J_y,
\end{eqnarray}
and it can be readily computed that $\zeta$ is invariant with respect to $z$ rotations. Therefore,
\begin{equation}
  R_j(\Omega)^{-1} \zeta(R_j(\Omega)\rho_j R_j(\Omega)^{-1})R_j(\Omega)=\zeta(\rho_j)  \qquad \forall\ \Omega \in SU(2).
  \label{eq:covcon}
\end{equation}
The map $\zeta$ is therefore invariant.
\end{proof}

Because a composition of invariant maps is also invariant, then any map of the form
\begin{equation}
  \xi(\rho)=q_0\rho+\sum_{k=1}^{2j} q_k\zeta^{\circ k}(\rho),
\label{eq:formco}
\end{equation}
where the values $q_i$ are real numbers and $\zeta^{\circ k}=\zeta\circ\zeta\circ\ldots\circ\zeta$, is also invariant with respect to $SU(2)$. It remains to show that any invariant map with respect to $SU(2)$ can be written on the form (\ref{eq:formco}).

To prove that every invariant map can be written as above, we use the Liouville representation of a superoperator \cite{havel}. Upon representing a $d\times d$ density matrix $\rho$ into a $d^2$ long column vector $\ketL{\rho}$ by stacking the columns of the density matrix, the action of any superoperator $\xi$ can be represented as a $d^2\times d^2$ matrix  $\mathcal{K}(\xi)$, such that
\begin{equation}
  \ketL{\xi(\rho)}= \mathcal{K}(\xi)\ketL{\rho}.
\end{equation}
This representation is necessarily basis dependent.  If a given process has Kraus operators $\{E_k\}$, the Liouville representation takes the form
\begin{equation}
  \mathcal{K}(\xi)=\sum_k E_k^*\otimes E_k
\end{equation}
where $*$ represents the complex conjugate with respect to the chosen basis.

\begin{lemma}
The Liouville representation of any map that is invariant with respect to $SU(2)$ has the form
\begin{equation}
  \mathcal{K}(\xi) = ( e^{-i\pi J_y} \otimes \mathbbm{1})  \sum_{k=0}^{2j}c_k\Pi_k (e^{i\pi J_y} \otimes \mathbbm{1}),
  \label{eq:proj}
\end{equation}%add   e^{-i\pi J_y} \otimes \mathbbm{1}(\cdot)e^{i\pi J_y} \otimes \mathbbm{1}
where $c_k\in\mathbb{C}$ and $\Pi_k$ is the projector into the $2k+1$ dimensional subspace of total angular momentum $k$.
\end{lemma}

\begin{proof}
The condition on a map to be invariant can thus be expressed as
\begin{equation}
  (R^*_j(\Omega) \otimes R_j(\Omega)) \mathcal{K}(\xi)= \mathcal{K}(\xi) (R^*_j(\Omega) \otimes R_j(\Omega)),
  \qquad \forall\ \Omega \in SU(2),
  \label{eq:cov}
\end{equation}
so the condition on the operators in this new representation is that $\mathcal{K}(\xi)$ must commute with $R^*_j \otimes R_j$.  If the group $SU(2)$ is represented in the $\{\ket{j,m}_{z}\}$ basis (i.e.\ the eigenstate of $J_z$), we have the relation $R^*_j(\Omega)= e^{-i\pi J_y} R(\Omega) e^{i\pi J_y}$, which implies
\begin{equation}
  (R_j(\Omega) \otimes R_j(\Omega))(e^{i\pi J_y} \otimes I) \mathcal{K}(\xi)(e^{-i\pi J_y} \otimes I) =(e^{i\pi J_y} \otimes I)  \mathcal{K}(\xi)(e^{-i\pi J_y} \otimes I)  (R_j(\Omega) \otimes R_j(\Omega))  \ \ \forall\ \Omega \in SU(2).
  \label{eq:cov2}
\end{equation}
We note that $(R_j(\Omega) \otimes R_j(\Omega))$ is the \emph{collective} representation of $SU(2)$ on two spin-$j$ systems.
A Clebsch-Gordon decomposition gives us the irrep of the group of all collective rotations $\mathcal{G}(R_j\otimes R_j)$ which take the form
\begin{equation}
  \mathcal{G}(R_j\otimes R_j)(\Omega) \simeq \bigoplus_{k=0}^{2j}R_k(\Omega) , \qquad \forall\ \Omega \in SU(2),
\end{equation}
where ``$\simeq$'' denotes ``unitarily equivalent'' and $R_k$ is the spin-$k$ irrep of $SU(2)$ which has multiplicity one.

Because we require $\mathcal{K}(\xi)$ to commute with $R_j(\Omega) \otimes R_j(\Omega)$ for any $\Omega \in SU(2)$, it must then commute with all irreps of $\mathcal{G}(R_j\otimes R_j)$.  By Schur's lemma, we have that
\begin{equation}
(e^{i\pi J_y} \otimes \mathbbm{1}) \mathcal{K}(\xi)(e^{-i\pi J_y} \otimes \mathbbm{1})  \simeq \bigoplus_{k=0}^{2j}c_k I_k
  =\sum_{k=0}^{2j}c_k\Pi_k,
\end{equation}%add   e^{i\pi J_y} \otimes \mathbbm{1}(\cdot)e^{-i\pi J_y} \otimes \mathbbm{1}
where $c_k\in\mathbb{C}$, $I_k$ is the $2k+1$ dimensional identity operator, and $\Pi_k$ is the projector into the $2k+1$ dimensional subspace of total angular momentum $k$.
\end{proof}

There are $2j+1$ independent projectors forming $\mathcal{K}(\xi)$.  To characterize every possible invariant mapping, we will thus also require $2j+1$ independent operators.

\begin{lemma}
The matrices $(\sum_{i} J^*_i \otimes J_i)^{k}$ for $0 \leqslant k \leqslant 2j$ are linearly independent and form a complete basis to represent any invariant map of the form~(\ref{eq:proj}).
\end{lemma}

\begin{proof}
If we rewrite equation~(\ref{eq:formco}) in the Liouville representation, the map becomes
\begin{equation}
  \mathcal{K}(\xi) = \sum_{k=0}^{n} \frac{q_{k}}{\lambda^k} \left(\sum_{i \in \{x,y,z\}} J^*_i \otimes J_i\right)^{k} .
\end{equation}
The hermitian matrices $(\sum_{i} J^*_i \otimes J_i)^{k}$ are diagonal in the same basis, and the eigenvalues of $(\sum_{i} J^*_i \otimes J_i)^{k}$ are $\nu_l^k$, where $\nu_l$ is an eigenvalue of $\sum_{i} J^*_i \otimes J_i$.

To find all of the eigenvalues $\{\nu_l\}$, expand the total angular momentum $\mathcal{J}^2=\sum_{i} \left(J_i \otimes I + I \otimes J_i\right)^2$ to get a new expression for $\sum_{i} J_i \otimes J_i$:
\begin{eqnarray}
  \label{eq:expansion}
  \sum_{i} J_i \otimes J_i &=& \frac{1}{2}(\sum_{i}(J_i \otimes I + I \otimes J_i)^2 - \sum_{i} J_i^2 \otimes I - I \otimes \sum_{i} J_i^2)\nonumber\\
 &=& \frac{1}{2}(\mathcal{J}^2 -J^2 \otimes I - I \otimes J^2) \nonumber\\
 &=& \frac{1}{2}(\mathcal{J}^2 - 2j(j+1)I)
\end{eqnarray}
With the relation $\sum_{i} J^*_i \otimes J_i=-e^{-i\pi J_y}\otimes I \left(\sum_{i} J_i \otimes J_i\right) e^{i\pi J_y}\otimes I$, we thus have that $\sum_{i} J^*_i \otimes J_i$ and $\sum_{i} J_i \otimes J_i$ will have the same eigenvalues up to a negative sign.

%>
From equation~(\ref{eq:expansion}), we find that
\begin{equation}
  \nu_l = -\frac{1}{2}\left(l(l+1) - 2j(j+1) \right),
\end{equation}
where $l(l+1)$ is the eigenvalue resulting from the vector addition of the two spin-$J$ systems, which implies that $l$ can range from $0$ to $2j$ and has multiplicity $2l+1$. There is $2j+1$ different values of $l$, so there is $2j+1$ different eigenvalues.  By the fundamental theorem of algebra, this implies that there exist no polynomial of degree $2j$ that has the eigenvalues $\{v_l\}$ as roots.  This implies the matrices $(\sum_{i} J^*_i \otimes J_i)^{k}$ for $0 \leqslant k \leqslant 2j$ are linearly independent.  By counting the number free parameters, we proved that the operators $(\sum_{i} J^*_i \otimes J_i)^{k}$ form a complete basis to represent any map represented by equation~(\ref{eq:proj}).
\end{proof}

Finally, to show that the equation~(\ref{eq:formco}) is a representation of all possible invariant maps on a spin-$j$ system, we need to show that the $q_i$'s are real.  First, note that
\begin{equation}
 \zeta(\rho_j)= \frac{1}{2 \lambda}(2J_z \rho_j J_z +J_{+}\rho_j J_{+}^{\dagger}+J_{-}\rho_j J_{-}^{\dagger}),
 \label{eq:formpm}
\end{equation}%changed \xi to \zeta
where $J_{\pm}=J_x \pm i J_y$.  Also, from equation~(\ref{eq:formpm}), the only contribution to $_{\hat{n}}\bra{j,-j}\xi(\ket{j,j}_{\hat{n} \hat{n}}  \bra{j,j})\ket{j,-j}_{\hat{n}}$ is from $q_{2j}\zeta^{\circ 2j}(\ket{j,j}_{\hat{n} \hat{n}}  \bra{j,j})$. Actually,
\begin{equation}
  _{\hat{n}}\bra{j,-j}\xi(\ket{j,j}_{\hat{n} \hat{n}}  \bra{j,j})\ket{j,-j}_{\hat{n}} =  q_{2j}\ _{\hat{n}}\bra{j,-j}\zeta^{\circ 2j}(\ket{j,j}_{\hat{n} \hat{n}}  \bra{j,j})\ket{j,-j}_{\hat{n}}.
  \label{eq:matrixelem}
\end{equation}
Because $ \xi(\rho_j)$ must be positive, then $q_{2j}$ is a non-negative real number. The contributions to $_{\hat{n}}\bra{j,-j+1}\xi(\ket{j,j}_{\hat{n} \hat{n}}  \bra{j,j})\ket{j,-j+1}_{\hat{n}}$ are from $q_{2j}\zeta^{\circ 2j}(\ket{j,j}_{\hat{n} \hat{n}}  \bra{j,j})$ and $q_{2j-1}\zeta^{\circ 2j-1}(\ket{j,j}_{\hat{n} \hat{n}}  \bra{j,j})$. Since $ \xi(\rho_j)$ is positive and $q_{2j}$ is real, then $q_{2j-1}$ must also be real. Note that $q_{2j-1}$ could be negative. By induction, it is easy to show that all the coefficients $q_{i}$ must be real.


\begin{thebibliography}{99}


\bibitem{GI05}  N. Gisin, S. Iblisdir, {\it quant-ph/0507118}.%add this reference
%
\bibitem{BIM05} E. Bagan, S. Iblisdir, and R. Munoz-Tapia, {\it Phys. Rev. A} {\bf 73} 022341 (2006).%add this reference
%
\bibitem{brst} S. Bartlett, T. Rudolph, R. Spekkens, and P. Turner, \textit{New J. Phys.} \textbf{8}, 58 (2006).
%
\bibitem{py} D. Poulin and J. Yard, \textit{New J. Phys.} \textbf{9}, 156 (2007).
%
\bibitem{brs} S. Bartlett, T. Rudolph, and R. Spekkens, \textit{Rev. Mod. Phys.} \textbf{79}, 555 (2007).
%
\bibitem{RBMC04} D. Rugar, R. Budakian, H. J. Mamin, and B. W. Chui, {\it Nature} (London), {\bf 430}, 329 (2004).%added this reference
%
\bibitem{SGBRZHY95} J. A. Sidles, J. L. Garbini, K. J. Bruland, D. Rugar, O. Z\"uger, S. Hoen, and C. S. Yannoni, {\it Rev. Mod. Phys.} {\bf 67}, 249 (1995).%added this reference
%
\bibitem{LCW98} D. A. Lidar, I. L. Chuang, and K. B. Whaley, {\it Phys. Rev. Lett.} {\bf 81} 2594 (1998).%added this reference
%
\bibitem{Ritter05} W. G. Ritter, {\it J. Math. Phys.} {\bf 46}, 082103 (2005).%add this reference
%
\bibitem{brs04} S. Bartlett, T. Rudolph, and R. Spekkens, \textit{Phys. Rev. A} \textbf{70}, 032321 (2004).
%
\bibitem{NielsenChuang97} M. A. Nielsen and I. L. Chuang, \textit{Phys. Rev. Lett.} \textbf{79}, 321 (1997).

\bibitem{MPRHorodecki99} M. Horodecki, P. Horodecki, and R. Horodecki, \textit{Phys. Rev. A 60}, \textbf{1888} (1999).

\bibitem{BowdreyOiShortBanaszekJones02} M. D. Bowdrey, D. K. L. Oi, A. J. Short, K. Banaszek, and J. A. Jones, \textit{Phys. Lett. A} \textbf{294}, 258 (2002).

\bibitem{Nielsen02}   M. A. Nielsen,  \textit{Phys. lett., A} \textbf{303}, 249 (2002).

\bibitem{EmersonAlickiZyczkowski05} J. Emerson, R. Alicki, and K. Zyczkowski, \textit{J. Opt. B: Quantum Semiclass. Opt.} \textbf{7} (2005).

\bibitem{rose} M. E. Rose, {\bf Elementary Theory of Angular Momentum},  (Wiley, New York, 1957).
%
\bibitem{G96} N. Gisin, {\it Phys.~Lett.~A} {\bf 210}, 151 (1996).
%
%\bibitem{joe} J. Emerson, \textit{unpublished}
%
%\bibitem{haake} F. Haake, \textit{Springer} (1992, 2001) ISBN 0172-7389.
%

\bibitem{KeylWerner99}  M. Keyl, and R. F. Werner, {\it J.Math.Phys.} {\bf 40}, 3283 (1999).

\bibitem{Ariano04}  G. M. D'Ariano, {\it J. Math. Phys.} {\bf 45}, 3620 (2004).

\bibitem{havel} T. F. Havel, \textit{J. Math. Phys.} \textbf{44}, $\#$2, 534 (2003).


\end{thebibliography}
\end{document}